\documentclass[11pt]{article}

\usepackage{amsmath}
\usepackage{fancyhdr}
\usepackage{amssymb}
\usepackage{amsthm}
\usepackage{graphicx}
\usepackage{varioref}
\usepackage{verbatim} 
\usepackage{multicol}
\usepackage{lmodern}
\usepackage{enumerate}
\usepackage[normalem]{ulem}
\usepackage{caption}
\usepackage{subcaption}
\usepackage[T1]{fontenc}
\usepackage[margin=1in]{geometry}

\usepackage{cite}
\usepackage{mathrsfs}

\usepackage{url}
\usepackage{xspace}
\usepackage{thm-restate}

\usepackage{mathpazo}

\usepackage{microtype}

\usepackage{tikz}
\usetikzlibrary{calc}

\usepackage{epstopdf}

\usepackage[colorlinks = true]{hyperref}
\usepackage{xcolor}
\definecolor{darkred}  {rgb}{0.5,0,0}
\definecolor{darkblue} {rgb}{0,0,0.5}
\definecolor{darkgreen}{rgb}{0,0.5,0}
\hypersetup{
  urlcolor   = blue,         
  linkcolor  = darkblue,     
  citecolor  = darkgreen,    
  filecolor  = darkred       
}

\usepackage{cleveref}
\crefname{lemma}{Lemma}{Lemmas}
\crefname{proposition}{Proposition}{Propositions}
\crefname{definition}{Definition}{Definitions}
\crefname{theorem}{Theorem}{Theorems}
\crefname{conjecture}{Conjecture}{Conjectures}
\crefname{corollary}{Corollary}{Corollaries}
\crefname{section}{Section}{Sections}
\crefname{appendix}{Appendix}{Appendices}
\crefname{figure}{Fig.}{Figs.}
\crefname{equation}{Eq.}{Eqs.}
\crefname{table}{Table}{Tables}

\newtheorem{theorem}{Theorem}
\newtheorem{lemma}[theorem]{Lemma}

\newtheorem*{conjecture*}{Conjecture}

\theoremstyle{definition}

\usepackage{authblk}

\newcommand{\ket}[1]{|#1\rangle}
\newcommand{\bra}[1]{\langle#1|}
\newcommand{\kb}[1]{|{#1}\rangle\!\langle{#1}|}

\DeclareMathAlphabet{\matheu}{U}{eus}{m}{n}

\DeclareMathOperator{\tr}{tr}

\newcommand{\Hilb}{{\matheu H}}

\newcommand{\state}{{\matheu D}}

\newcommand{\sop}[1]{{\mathcal #1}}

\begin{document}

\title{No-Go Bounds for Quantum Seals}
\author[1]{Shelby Kimmel}
\affil[1]{Department of Computer Science,
Middlebury College, \texttt{skimmel@middlebury.edu}}
\author[2]{Shimon Kolkowitz}
\affil[2]{Department of Physics, University of Wisconsin-Madison, \texttt{kolkowitz@wisc.edu}}

\date{}

\maketitle

\begin{abstract}

We investigate the feasibility of quantum seals. A quantum seal is a state
provided by Alice to Bob along with information which Bob can use to make a measurement, 
``break the seal,'' and read the classical message stored inside.
There are two success criteria for a seal: the probability Bob can
successfully read the message without any
further information from Alice must be high, and if Alice asks for the state back from Bob, the
probability Alice can tell if Bob broke the seal without permission must be
high. We build on the work of [Chau, PRA 2007], which gave optimal bounds on these criteria, 
showing that they are mutually exclusive for high probability. We weaken the assumptions of this previous work by providing Bob with only a classical description of a prescribed measurement, rather than classical descriptions of the possible seal states. We show that this weakening does not affect the bounds but does simplify the analysis. We also prove upper and lower bounds on an alternative operational metric for measuring the success criteria.
\end{abstract}



\section{Introduction}
\label{Intro}
With a quantum seal, a sender hopes to communicate classical information using a quantum state and also verify whether that information has been extracted.
This functionality is called a ``seal'' 
in analogy with the impressions made in wax or clay that have been 
used to ensure the integrity of letters and packages for thousands of years \cite{HistoryofSeals}. 
Alice wishes to give Bob a message that is only to be opened and read by Bob at a later date when an
agreed upon set of conditions has been met. For example, the conditions might be
``Do not open the message until the third night of Hannukah,'' or 
``Only open the message if instructed to do so by Alice.''  
To this end, Alice gives Bob a quantum state and a classical description of a quantum
measurement. If Bob makes the measurement described by Alice, with high
probability he will extract the classical value of the message. Bob should be able
to make this measurement without further information from Alice. 
However, at any time before the agreed upon conditions have been met, 
Alice can request the state back, and she would like to detect with high probability whether Bob has
cheated and read the message, thereby ``breaking the seal'' prematurely.  
A seal therefore has two success criteria:
1) the probability Bob can successfully read the message if he follows Alice's instructions must be high, 
and 2) the probability Alice can tell if Bob broke the seal without permission must be high. 

Quantum seals were introduced by Bechmann-Pasquinucci \cite{QSeal1}, giving
rise to a vibrant discussion into their feasibility under
different considerations
\cite{QSeal2,QSeal3,QSeal3a,QSeal4,QSeal4b,QSeal5,QSeal6,QSeal7,QSeal8}. The upshot is that without limitations on Bob's
possible measurements, if Alice wants Bob to be able to learn her complete
classical message with high probability, then she
will not be able to detect his breaking of the seal with high probability\footnote{In past literature, seals where Alice and Bob needed to be able to succeed in their tasks with certainty were called ``perfect'' seals while ``imperfect'' seals allowed for only probabilistic success. When more than a single bit was encoded the seal was called a ``string seal.'' Here we do not bother with these distinctions, as our results are general and apply to all of these cases.}. However, this does not necessarily mean that quantum seals are entirely useless; for example, to send a very long message Alice can design a quantum seal that gives her a very high probability to catch Bob cheating, while also giving Bob a low but still finite probability of extracting the correct message that is much higher than his vanishingly small chance of randomly guessing it \cite{QSeal5}.

Quantum seals are related to other quantum cryptographic protocols with no-go results such as 
quantum bit commitment \cite{Qbitcommit,Qbitcommit2,BitCommitRudolph} and quantum one-time memories \cite{goldwasser2008one}. In bit commitment, Alice first provides a quantum
state, which is her commitment to a bit, and then she might later be asked to reveal her bit. Note this is in contrast to quantum seal protocols, in that once Alice has provided the seal to
Bob, she is never asked to provide further information. It has been shown that quantum bit commitment can be built from quantum seals, and so is a strictly weaker primitive \cite{sealVcommit}.

On the other hand, a straightforward argument shows that a one-time memory (OTM) can be used to create a nearly perfect seal. An OTM is a
device that contains two messages $s,t\in \{0,1\}^n$. Once one message is
read, the other is destroyed. To create a seal from an OTM, Alice
would set $s$ to be the message and $t$ to be a random string. Bob could
therefore learn $s$ perfectly without further input from Alice. However, once he reads $s$, he destroys $t$. 
To test if Bob had read the message, Alice could ask for the OTM back, try to learn $t$, and if
it was inaccessible or not the string she stored, she would know Bob had
cheated. Bob could try to
make a new OTM to give back to Alice as a fake, but he would not
know what string to store as $t$, so his guess would be inaccurate with high
probability.  Since seals are weaker than OTMs, a no-go for quantum
 seals does not necessarily follow from the no-go for quantum OTMs. However, we note that under certain physically realistic assumptions
 (such as no entangling operations), quantum OTMs can exist
 \cite{liu2014building}, which implies the existence of quantum seals under
 similar restricted scenarios.

In prior work, Chau \cite{QSeal5} proved optimal bounds on the success criteria of quantum seals in the case that Bob has knowledge of the different quantum states that he might receive from Alice. In this work, we weaken Bob's advantage by \textit{not} giving him a description of the seal states, and instead giving him only a classical description of a quantum measurement and instructions for how to associate the possible outcomes with classical messages. Alice promises Bob that if he uses the specified measurement on the seal state, he will obtain the correct message with some guaranteed probability. She does not provide Bob with any further information as to how to implement the quantum measurement; Bob's choice of implementation affects the probability that his measurement can be detected by Alice\footnote{Our setup is therefore similar to that in \cite{QSeal3}, but with the possibility of sealing a message that is longer than a single bit.}. Since Bob does not have information about the underlying states, it is potentially harder for him to design a cheating strategy.

Even with this restriction on Bob's information, we show that Bob can still read the message with high probability without detection. In particular, when we consider the case that Alice never wants to falsely accuse Bob of cheating, which is the same metric used by Chau \cite{QSeal5}, we achieve the same bounds. Thus it would seem that restricting Bob's information in this way does not have a significant effect, while our analysis is somewhat simplified relative to Chau's. 

We additionally examine an alternative operational metric for Alice's success in detecting whether Bob has broken the seal: we look at Alice's probability of detecting Bob cheating if she makes the optimal measurement to distinguish between the broken and unbroken seal states. This measurement may sometimes cause her to falsely accuse Bob of breaking the seal when he hasn't, but gives her a higher probability of detecting Bob's measurement when he has broken the seal\footnote{This is in contrast to the metric used by Chau \cite{QSeal5}, in which Alice makes a measurement that would never result in a false accusation of cheating against Bob, but which as a result is suboptimal for distinguishing between a broken and unbroken seal state.}. Under this metric we show that Alice's chances of detecting Bob, though improved, are still not good.

In order to obtain some of our results, in Sec.~\ref{Lemma} we prove a variation on 
the Gentle Measurement Lemma \cite{wilde2013quantum,winter1999coding,Watrous2018} which may be of independent interest.
The Gentle Measurement Lemma states that if a measurement outcome occurs with high probability, and
if that outcome is measured, then the state after the measurement is difficult
to distinguish from the original state. We
extend the Gentle Measurement Lemma to show that the state after measurement is difficult to distinguish from the original state with high probability even if the outcome of the measurement is not known.


\section{Preliminaries}

\subsection{Notation and Quantum Measurement}\label{sec:quantumBackground}
\label{Notation}

We use $\Hilb$ to denote a Hilbert space, and $\state(\Hilb)$ to denote the set
of positive linear operators acting on $\Hilb$ with trace one; $\state(\Hilb)$
is the set of density matrices on $\Hilb$. For $N\in \mathbb
{N}$, we let $[N]=\{1,\dots,N\}.$ $\mathcal{I}_A$ denotes the identity operator on $\Hilb_A$, but
we drop the subscript if clear from context.

A quantum measurement is described by a positive operator value measure (POVM).
A POVM is a set of operators $\{\mathcal{E}_i\}_{i\in [N]}$ acting on a Hilbert space
$\Hilb$ such that $\sum_{i=1}^N\mathcal{E}_i=\mathcal{I}$ and $\mathcal{E}_i\succeq 0$. Given a state $\rho\in
\state(\Hilb)$, the probability of measuring outcome $i$ is $\tr(\mathcal{E}_i\rho)$.

There are an infinite number of ways to implement a given POVM
\cite{Kraus83,braunstein1988quantum} (see \cite{peres2006quantum} for a nice
description of methods to implement a POVM) and the implementation affects the
state the system is left in after the measurement. Here, we think of the implementation as a two step process. In the first step, we apply what we call the \textit{standard implementation}: if $\rho$ is measured and outcome $i$ is obtained, the state collapses to $\sqrt{\mathcal{E}_i}\rho\sqrt{\mathcal{E}_i}/\tr(\mathcal{E}_i\rho )$. In the second step, a completely positive trace preserving (CPTP) map is applied to the resultant state, and this map can depend on the outcome $i$ of the first step.

If $\rho$ is measured using a POVM $\{\mathcal{E}_i\}$, and if the outcome of the measurement is not known, the post-measurement system is described by the density matrix $\sum_i\sqrt{\mathcal{E}_i}\rho\sqrt{\mathcal{E}_i}$ \cite{nielsen2002quantum}.

The trace distance between quantum states $\rho,\sigma\in \state(\Hilb)$ is
$\frac{1}{2}\|\rho-\sigma\|_1$, where $\|A\|_1 = \tr(\sqrt{A^\dagger A})$. The
trace distance has an important operational meaning: given a state promised to be either
$\rho$ or $\sigma$ with equal probability, the maximum probability of correctly
distinguishing which state one has is $\frac{1}{2}+\frac{1}{4}\|\rho-\sigma\|_1.$

Finally, let
\begin{align}\label{eqn:rho}
Z(x)=\left(\begin{array}{cc}
x & 0\\
0 & 1-x
\end{array}\right).
\end{align}

\subsection{Set-Up}
\label{TheProblem}
We assume the message protected by the seal takes a value $m\in[M]$ for $M\geq 2$.
Alice encodes the message into a pure state $\ket{\psi_m}\in \Hilb_{A}\otimes
\Hilb_{B}$ where $\Hilb_A$ refers to System $A$, the part of the system
that Alice retains, while $\Hilb_B$ refers to System $B$, the part of the system Alice
gives to Bob as the sealed message. (If Alice instead would prefer to encode into a mixed
state, we can always purify the state without loss of generality. See e.g.
\cite{nielsen2002quantum}.) Thus Alice sends to Bob the state
$\rho_m\in \state(\Hilb_B)$, where $\rho_m$ is the reduced density matrix of
$\ket{\psi_m}$ on $\Hilb_B$.

In addition to the state $\rho_m$, Alice gives Bob a classical description of a POVM $\vec{\mathcal{E}}=\{\mathcal{E}_{i,j}\}_{i\in[M],j\in [M_i]}$. She promises Bob that if he
performs the POVM $\vec{\mathcal{E}}$  on $\rho_m$, he will get an outcome $(i,j)$ such that $i=m$ with probability at least $p$. This promise implies
\begin{align}\label{eq:bound}
\sum_{j\in [M_m]}\tr\left(\mathcal{E}_{m,j}\rho_m\right)\geq p.
\end{align}

After giving System $B$ to Bob, Alice may ask for it to be returned to her at any point. Bob's goal in our scenario is to make a measurement on $\rho_m$ that allows him to determine $m$, but when he returns the system to Alice, she can not detect his measurement. Alice's goal is to design a state $\rho_m$ and a POVM $\vec{\mathcal{E}}$, with the properties described above, such that Bob can not learn $m$ without significantly altering $\ket{\psi_m}$, so that when Alice asks for System $B$ to be returned, she can reliably determine whether Bob has cheated.

We use two metrics to judge Alice's success in detecting Bob's potential
breaking of the seal. First, suppose Alice has no prior knowledge as to whether Bob has cheated, so her a priori belief is that she either has her original state or a broken state with equal probability. We call $p_{\textrm{dist}}$ her maximum success probability in distinguishing whether or not Bob made a measurement that allowed him to learn the message with probability at least $p$. Second, suppose Alice would like to detect
if Bob made a measurement that allowed him to learn the message with probability at least $p$, but never wants to falsely accuse Bob
of cheating if he didn't actually read the message. We call $p_{\textrm{NFP}}$ (NFP
for ``no false positives'') her maximum success probability in this task. This metric was previously used by Chau to quantify Alice's ability to detect cheating, but without the ``no false positives'' interpretation \cite{QSeal4b,QSeal5}. If Alice chooses her message state $\ket{\psi_m}$ from some distribution of states, we take $p_{\textrm{dist}}$ and $p_{\textrm{NFP}}$ to be her success probabilities averaged over her choice of state.

For the quantum seal to behave as desired, we would like to have a protocol
in which $p$ is large, and $p_{\textrm{dist}}$ or $p_{\textrm{NFP}}$ are large, so that Bob can read the message correctly with high probability, 
but also Alice can detect if he read it or not. However, we show that if $p$
is large, Bob can implement a POVM on $\rho_m$ in such a way that $p_{\textrm{dist}}$ and
$p_{\textrm{NFP}}$ will be small.

\subsection{A Variation of the Gentle Measurement Lemma}
\label{Lemma}

The Gentle Measurement Lemma \cite{wilde2013quantum,winter1999coding, Watrous2018} says that if a state has high overlap with a POVM operator and that outcome is measured, then the post-measurement state will not differ considerably from the pre-measurement state. While there are several different formulations of the Lemma, the one that is most relevant to this work is formulated as follows: given a POVM operator $\sop E_j$ on $\Hilb$
and a state $\rho\in D(\Hilb)$, where $\tr(\mathcal{E}_j\rho)\geq 1-\epsilon$, then
\begin{align}\label{eq:gentle}
\left\|\rho-\sqrt{\mathcal{E}_j}\rho\sqrt{\mathcal{E}_j}\right\|_1\leq 2\sqrt{\epsilon}.
\end{align}

Here we consider a variation of this Lemma for the case where an initial state $\rho$ is measured with a POVM, $\vec{\mathcal{E}}$, where one of the POVM operators $\sop E_j$ has high overlap with $\rho$, \textit{but where the outcome of the measurement is unknown.} 
This scenario occurs, for example, if Alice knows that Bob has made a measurement, but he hasn't told her what the outcome of his measurement is. When the outcome of a measurement is unknown, the best description of the post-measurement state is given by a probabilistic combination of all possible outcomes of the measurement. In particular, as discussed in \cref{sec:quantumBackground}, if the POVM $\{\sop E_i\}$ is applied to a state $\rho$, but the outcome is unknown, the density matrix of the resultant state is given by $\sum_i\sqrt{\mathcal{E}_i}\rho\sqrt{\mathcal{E}_i}.$ 

We show that the trace distance between the original state and the post measurement state is still small, even in the case that the outcome of the measurement is unknown. In particular, we prove the following:
\begin{lemma}\label{lemm:newGentle}
Given a state $\rho\in D(\Hilb)$ and a POVM $\{\mathcal{E}_i\}$ on $\Hilb$ such that there is one POVM
operator $\mathcal{E}_j$ where $\tr(\mathcal{E}_j\rho)\geq 1-\epsilon$, then the trace distance between the pre-measurement state and the post-measurement state when the measurement outcome is unknown is bounded as
\begin{align}
\left\|\rho-\sum_i\sqrt{\mathcal{E}_i}\rho\sqrt{\mathcal{E}_i}\right\|_1\leq 2\sqrt{\epsilon}+\epsilon.
\end{align}
\end{lemma}
\begin{proof}
Using the triangle inequality:
\begin{align}\label{eq:gentleNew}
\left\|\rho-\sum_i\sqrt{\mathcal{E}_i}\rho\sqrt{\mathcal{E}_i}\right\|_1&\leq \left\|\rho-\sqrt{\mathcal{E}_j}\rho\sqrt{\mathcal{E}_j}\right\|_1+\sum_{i\neq j} \left\|\sqrt{\mathcal{E}_i}\rho\sqrt{\mathcal{E}_i}\right\|_1\nonumber\\
&\leq 2\sqrt{\epsilon}+\sum_{i\neq j} \left\|\sqrt{\mathcal{E}_i}\rho\sqrt{\mathcal{E}_i}\right\|_1,
\end{align}
where we've applied \cref{eq:gentle}. Because $\sqrt{\mathcal{E}_i}\rho\sqrt{\mathcal{E}_i}$ is positive semidefinite, and using properties of the trace,
\begin{align}\label{eq:gentle1}
\sum_{i\neq j} \left\|\sqrt{\mathcal{E}_i}\rho\sqrt{\mathcal{E}_i}\right\|_1=\sum_{i\neq j} \tr(\sqrt{\mathcal{E}_i}\rho\sqrt{\mathcal{E}_i})=\sum_{i\neq j} \tr(\mathcal{E}_i\rho)=\tr\left(\sum_{i\neq j}\mathcal{E}_i\rho\right).
\end{align}
Now since $\sum_i\mathcal{E}_i=\mathcal{I}$ and $\tr(\rho)=1$,
\begin{align}\label{eq:gentle2}
\tr\left(\sum_{i\neq j}\mathcal{E}_i\rho\right)=1-\tr\left(\mathcal{E}_j\rho\right)\leq \epsilon.
\end{align}
Combining \cref{eq:gentle1} and \cref{eq:gentle2} and plugging into \cref{eq:gentleNew}, we have
\begin{align}
\left\|\rho-\sum_i\sqrt{\mathcal{E}_i}\rho\sqrt{\mathcal{E}_i}\right\|_1\leq 2\sqrt{\epsilon}+\epsilon.
\end{align}
\end{proof}
\section{A Naive Approach that Fails}
\label{Naive}

In this section, we present a straightforward strategy that seems promising\footnote{Indeed, this strategy resembles the initial proposal for a quantum seal by Bechmann-Pasquinucci \cite{QSeal1}.}, yet ultimately fails.

Let $M=2$, so Alice wants to encode a binary message. Suppose Alice creates a state on $3q$ qubits that she plans to send entirely to Bob; that is, she sets $\Hilb_A=1$, and $\Hilb_B=\mathbb{C}^{2^{3q}}$. She chooses $\sigma,\tau\in S_{3q}$ uniformly at random, where $S_n$ is the symmetric
group of degree $n$. Then Alice gives one of the following states to Bob, depending on whether she wants the message to be ``1'' or ``2'':
\begin{align}
\ket{\psi_{1}}&=U_\sigma\ket{0}^{\otimes 2q}\ket{+}^{\otimes q},\nonumber\\
\ket{\psi_{2}}&=U_\tau\ket{1}^{\otimes 2q}\ket{+}^{\otimes q},
\end{align}
where $U_\sigma$ (respectively $U_\tau$) is a unitary that
acts on a Hilbert space of $3q$ qubits, and permutes the qubit registers
according to the permutation $\sigma$ (resp. $\tau$).

Alice tells Bob to measure each qubit using the POVM $\{\kb{0},\kb{1}\}$ (i.e. the standard basis projective measurement), and if the number of $0$ outcomes is at least $3q/2$, he should decide $m=1$, and
otherwise, he should decide $m=2$.
If Bob uses the standard implementation of Alice's POVM, he will perform a projective measurement, and he will be able to read the message perfectly, since the the number of $0$ outcomes will be at most $q$ when $m=1$, and at least $2q$
when $m=2.$

After the standard implementation, Bob is left with a standard basis state, which will be nearly orthogonal to the
original state. Thus if Alice asked for the state back, she would with high
probability be able to detect Bob's measurement. 

Bob could try to disguise his measurement by applying a CPTP map to alter his state after measurement.
Let's assume, without loss of generality, that $m=1$, and also that Bob knows that Alice
originally sent a state of the form $U_\sigma\ket{0}^{\otimes
2q}\ket{+}^{\otimes q}$ for some $\sigma\in S_{3q}$. (This extra information can only help Bob.) Bob can replace
qubits in registers where he got outcome $1$ with states $\ket{+}$, to try to
make his state closer to Alice's original state. However with extremely high
probability in the limit of large $q$ (using e.g. Hoeffding's inequality
\cite{Hoeff63}), he will measure $0$'s in about half of the registers that originally contained
the state $\ket{+}$. For large $q$, Bob has a vanishingly small probability of correctly
guessing where these ``false 0'' registers are, and if he guesses
incorrectly, it will make his overlap worse. Thus, there is very little Bob
can do to recover from the measurement; the seal has been broken, and Alice will detect his measurement. 

So why does this protocol fail? While Alice told Bob that he should measure in
the standard basis, Bob can instead use Alice's instructions to make a different but related measurement. He measures using the standard implementation of the POVM $\{\Pi_1,\Pi_2\}$ where $\Pi_1$ is the projector onto standard
basis states whose strings have more than $3q/2$ zeros, and $\Pi_2$ is the
projector onto the remaining standard basis states. Bob has thereby combined all of the measurement operators that correspond to a given outcome into a single measurement operator.
For any choice of $\sigma,\tau\in S_{3q}$,
$\Pi_1\ket{\psi_{1}}=\ket{\psi_{1}}$ and
$\Pi_2\ket{\psi_{2}}=\ket{\psi_{2}}$. Thus Bob can deterministically distinguish the value of the message without disturbing the state and breaking the seal, and Alice will be completely unaware of his measurement.

In the next section, we show that there is always a way for Bob to cheat in a
manner similar to this, as long as Alice wants Bob to be able to read the message
with high probability.

\section{No-Go For Quantum Seals}
\label{sec:noGo}

We will show that a good strategy for Bob is to apply the standard implementation of the POVM $\{F_i\}_{i\in [M]}$, for 
\begin{align}\label{eq:POVMcombine}
 \mathcal{F}_i=\sum_{j\in [M_i]}\mathcal{E}_{i,j},
 \end{align}
where $\mathcal{E}_{i,j}$ are the elements of Alice's recommended POVM. If outcome $\mathcal{F}_i$ occurs, Bob decides the message is $i$. Averaged over Bob's outcome, the full state on $\Hilb_A\otimes \Hilb_B$ after
measurement is (see Section \ref{sec:quantumBackground})
\begin{align}\label{eq:postMeas}
\sum_{i=1}^{M}\mathcal{I}_A\otimes \sqrt{\mathcal{F}_i}\kb{\psi_m} \mathcal{I}_A\otimes \sqrt{\mathcal{F}_i}.
\end{align}

Now if Alice asks for Bob to return his system, and he did not make a measurement, she will have the state $\kb{\psi_m}$. If he did make the measurement using the POVM in \cref{eq:POVMcombine}, the state will be that in \cref{eq:postMeas}. 

We first bound $p_{\textrm{dist}}$. Alice's goal is to determine
which state she possesses. We assume Alice knows that if Bob made a
measurement, he measured using the standard implementation of the POVM in \cref{eq:POVMcombine},
as this information can only help her. Then using \cref{lemm:newGentle}, we have
\begin{align}\label{eq:pGuessFinal}
p_{\textrm{dist}}\leq&\frac{1}{2}+\frac{1}{4}\left\|\kb{\psi_m}-\sum_{i=1}^{M}\mathcal{I}_A\otimes \sqrt{\mathcal{F}_i}\kb{\psi_m} \mathcal{I}_A\otimes \sqrt{\mathcal{F}_i}\right\|_1\leq\frac{1}{2}+\frac{1}{4}\left(2\sqrt{1-p}+1-p)\right),
\end{align} 
which we plot in \cref{fig:guess}.
\begin{figure*}[t]
\begin{center}
\includegraphics{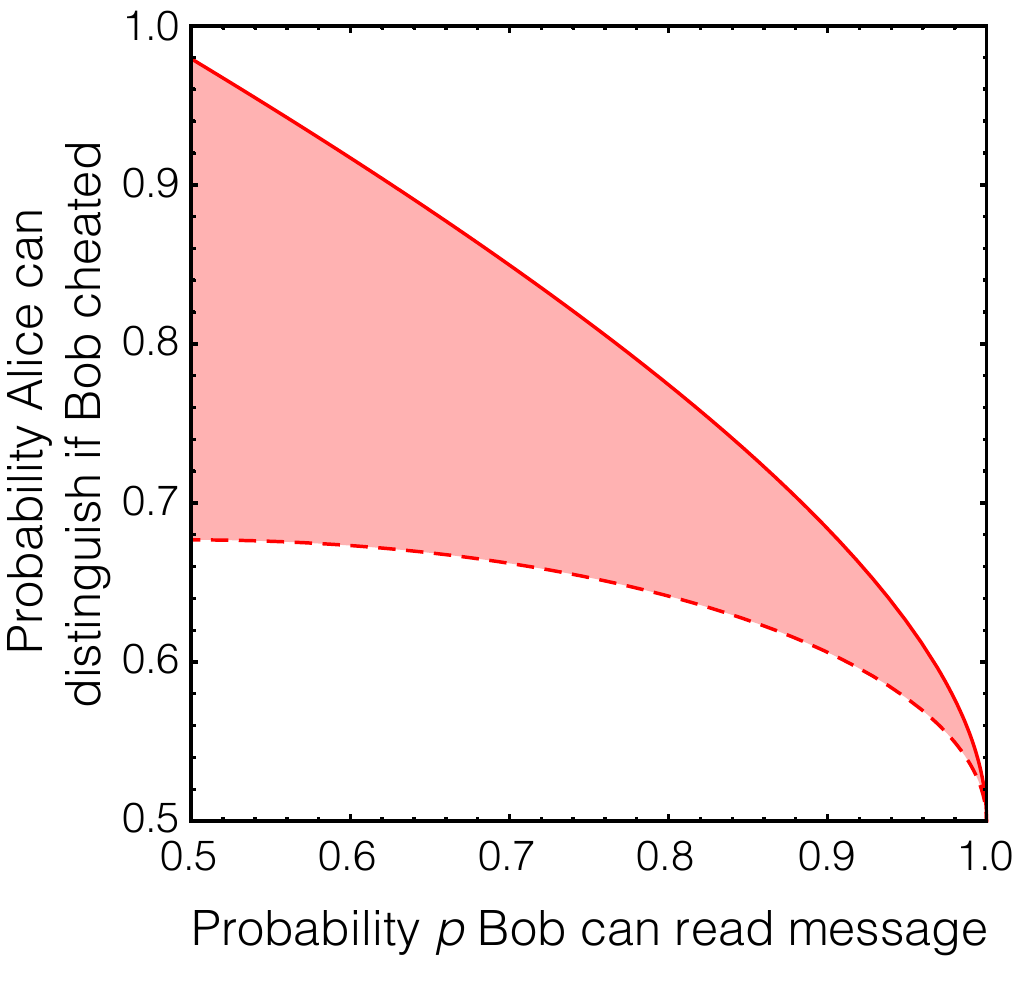}
\caption{An upper bound on the probability Alice can correctly distinguish whether Bob has cheated, $p_{\textrm{dist}}$, is plotted as a function of the probability $p$ that Bob can successfully read the sealed message if he follows Alice's instructions (solid red line,  Eq.~\ref{eq:pGuessFinal}). This bound holds for any message length ($M\geq 2$). A lower bound on the achievable $p_{\textrm{dist}}$ in the case of a single bit message ($M=2$), using the protocol described in Section ~\ref{Achieve}, is also shown (dashed red line, Eq.~\ref{eq:distAchieve}). The shaded portion represents what may be achievable for $M=2$ using a different protocol.} 
\label{fig:guess}
\end{center}
\end{figure*}

Next we bound $p_{\textrm{NFP}}.$
If Bob is honest, when he returns System $B$ to Alice, she will have the state $\ket{\psi_m}$. Therefore, Alice needs a two-outcome POVM, such that one outcome will never occur if Bob is honest. The POVM that achieves this is $\{\mathcal{I}_{AB}-\kb{\psi_m},\kb{\psi_m}\}$, where the first outcome will only be observed if Bob is dishonest. Thus
\begin{align}
p_{\textrm{NFP}}=&\tr\left(\left(\mathcal{I}_{AB}-\kb{\psi_m}\right)\left(\sum_{i\in [W]}\left(\mathcal{I}_A\otimes \sqrt{\mathcal{F}_i}\kb{\psi_m} \mathcal{I}_A\otimes \sqrt{\mathcal{F}_i}\right)\right)\right)\nonumber\\
=& 1-\sum_{i\in [M]}|\bra{\psi_m}\mathcal{I}_A\otimes\sqrt{\mathcal{F}_i}\ket{\psi_m}|^2
\end{align}

Now
\begin{align}
\bra{\psi_m}\mathcal{I}_A\otimes \sqrt{\mathcal{F}_i}\ket{\psi_m}\geq\bra{\psi_m}\mathcal{I}_A\otimes \mathcal{F}_i\ket{\psi_m}=\tr(\mathcal{F}_i\rho_m)
\end{align}
because the eigenvalues of $\mathcal{F}_i$ are between 0 and 1. Plugging in and using Cauchy-Schwarz, we have
\begin{align}\label{eq:upperPNFP}
p_{\textrm{NFP}}\leq 1-\tr(\mathcal{F}_m\rho_m)^2-\frac{(\sum_{i\in [M]\setminus m}\tr(\mathcal{F}_i\rho_m))^2}{M-1}=1-\tr(\mathcal{F}_m\rho_m)^2-\frac{(1-\tr(\mathcal{F}_m\rho_m))^2}{M-1}.
\end{align}
For values of $\tr(\mathcal{F}_m\rho_m)\geq 1/M$, this expression is decreasing in $\tr(\mathcal{F}_m\rho_m)$, so since $\tr(\mathcal{F}_m\rho_m)\geq p$, we have
\begin{align}\label{eq:NFPBounds}
 p_{\textrm{NFP}}\leq 1-p^2-\frac{(1-p)^2}{M-1}.
 \end{align} 
 Bounds on $p_{\textrm{NFP}}$ for several values of $M$ are shown in \cref{fig:NFP}. This formula is identical to Chau's bound \cite[Eq. 33]{QSeal5}\footnote{In Chau's case, he deals with two probabilities, $p$ and $p_\textrm{max}$, which result from Bob having knowledge of the seal state. In our case Bob has no knowledge of the seal state, so effectively $p=p_\textrm{max}$.}. Thus unfortunately we find Alice can not boost her success of detecting Bob's breaking of the seal by withholding information about the seal states.

\begin{figure*}[t]
\begin{center}
\includegraphics{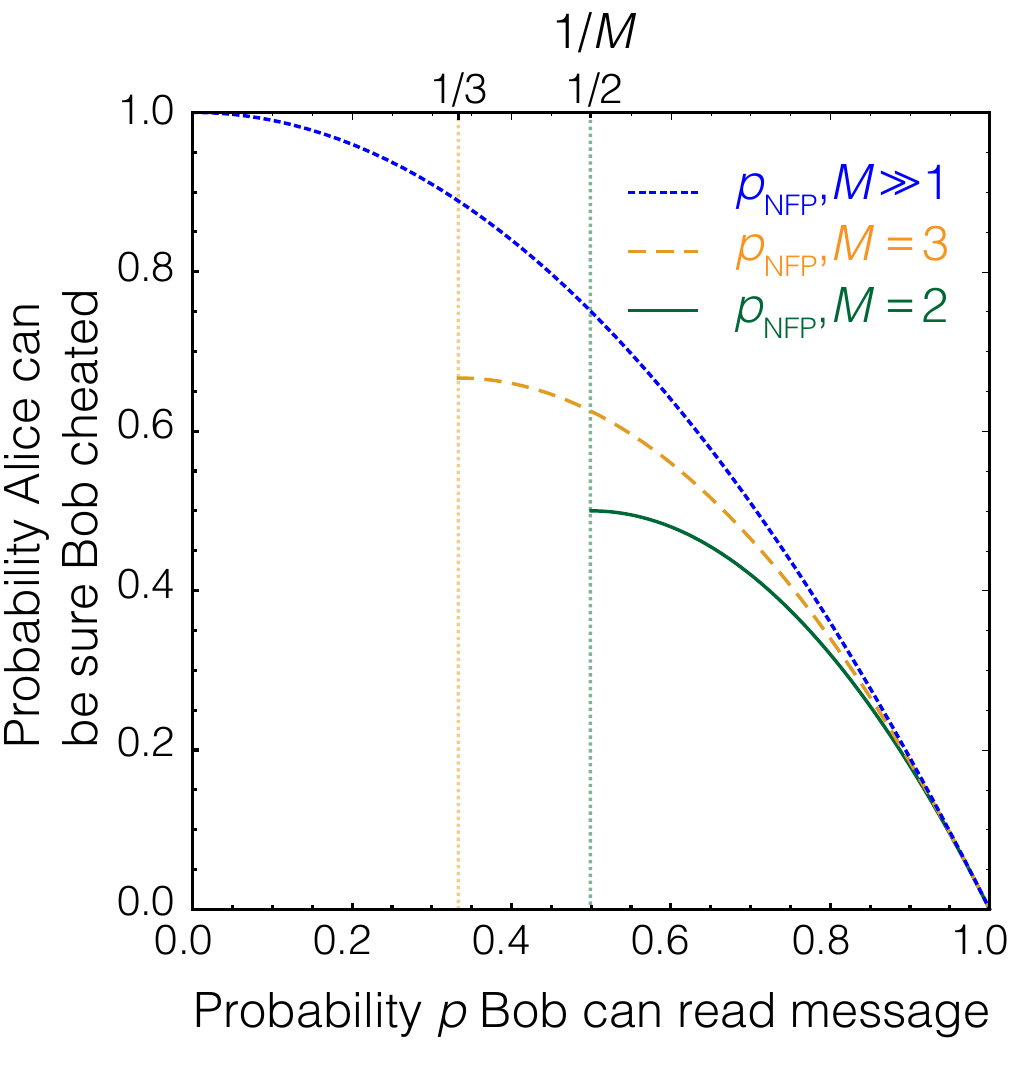}
\caption{Tight bounds on the probability Alice can know with certainty that Bob has cheated, $p_{\textrm{NFP}}$, are plotted as a function of the probability $p$ that Bob can successfully read the message if he follows Alice's instructions, for different ranges $M$ of possible message values (Eq.~\ref{eq:NFPBounds}). Without any information Bob can always correctly guess the value of the message with probability $1/M$, and we therefore plot each bound over the range $p\geq1/M$ (dashed vertical lines).}
\label{fig:NFP}
\end{center}
\end{figure*}

By either metric ($p_{\textrm{dist}}$ or $p_{\textrm{NFP}}$), we see that there is a trade off. If Alice wants Bob to be able to read the message with probability close to 1, then she will not be able to detect with high probability whether he has broken the seal.


\section{Achievability}
\label{Achieve}

In order to lower bound the achievability of the upper bounds of \cref{sec:noGo}, it is sufficient to analyze the case that Alice not only provides a classical description of a POVM (and information about how to associate outcomes with classical messages), but additionally gives Bob information about the quantum seal states. Since this additional information can only help Bob, it will allow us to put lower bounds on $p_{\textrm{dist}}$ and $p_{\textrm{NFP}}$.

This is precisely the scenario that Chau analyzed, and he showed that his bounds are achievable \cite{QSeal5}. Since our $p_{NFP}$ is the same as his bounds, this implies that our $p_{NFP}$ bound is optimal (see Fig \ref{fig:NFP}). 
It is interesting to note that when $M\gg 1$, for which Bob will have a negligibly small probability of randomly guessing the correct message, quantum seals still offer some rather non-trivial capabilities. For example, Alice can create a quantum seal where Bob has a $10\%$ chance of correctly reading the entire message, but Alice can detect if Bob has cheated with $99\%$ probability and no false positives. 

We next investigate whether Alice can achieve the bounds of \cref{sec:noGo} for the metric $p_{\textrm{dist}}$ for the case of a single bit message ($M=2$). We consider a specific strategy for Alice and determine the probability with which she can detect Bob's cheating. This allows us to put a lower bound on $p_{\textrm{dist}}$, but the specific strategy we choose might not be optimal.

Consider the case that $\Hilb_A=1$, $\Hilb_B=\mathbb{C}^2$, and $M=2$,
\begin{align}
\ket{\psi_1}=&\sqrt{p}\ket{0}+e^{i\phi}\sqrt{1-p}\ket{1},\\
\ket{\psi_2}=&\sqrt{1-p}\ket{0}+e^{i\phi}\sqrt{p}\ket{1},
\end{align}
for $\phi$ that Alice has chosen uniformly at random from $[0,2\pi]$, and 
$\vec{\mathcal{E}}=\{\mathcal{E}_1,\mathcal{E}_2\}$ where
\begin{align}\label{eq:POVM}
\mathcal{E}_1=\left(\begin{array}{cc}
1 & 0\\
0 &0
\end{array}\right),
\qquad
\mathcal{E}_2=\left(\begin{array}{cc}
0 & 0\\
0 &1
\end{array}\right).
\end{align}
Alice provides all of this information to Bob, as well as telling him the value of $p$ (we assume $p>1/2$). As discussed above, since we are putting a lower bound on $p_{\textrm{dist}}$, it is acceptable for Alice to give Bob this extra information about the seal states, since that information can only help him.

Alice's success probability $p_{\textrm{dist}}$ is always analyzed in the case that Bob makes a measurement that obtains the correct outcome with probability at least $p$. But note that the unique optimal POVM for distinguishing $\ket{\psi_1}$ and $\ket{\psi_2}$ is the POVM $\mathcal{E}_1$ and $\mathcal{E}_2$ from \cref{eq:POVM} \cite{blaquiere1987information}, and this measurement only succeeds with probability $p$. Thus Bob must use this POVM to achieve the desired success probability.

Bob wants to minimize $p_{\textrm{dist}}$ on average over Alice's choice of $\phi$. Using the fact that the trace norm is equal to the Euclidean norm on the Bloch sphere \cite{nielsen2002quantum}, and using the fact that the position with the smallest distance on average to any point on a circle is in the center of the circle, we have that Bob would ideally like to return $Z(p)$ if $m=1$, and $Z(1-p)$ if $m=2$ (see \cref{eqn:rho}). 
Luckily for Bob, if he simply returns the the result of his standard implementation of $\sop E$ (either the state $\ket{0}$ or $\ket{1}$) to Alice, on average over his measurement outcomes, he will return precisely $Z(p)$ in the case that $m=1$, and $Z(1-p)$ when $m=2.$ Thus Bob's optimal strategy for minimizing $p_{\textrm{dist}}$ is to do a projective measure in the standard basis and return the outcome state if Alice requests it.

We can now bound $p_{\textrm{dist}}$ (assuming Alice knows Bob implements the optimal strategy, since there is no reason for Bob to use any other strategy):
\begin{align} \label{eq:distAchieve}
p_{\textrm{dist},(M=2)}\geq\frac{1}{2}+\frac{1}{2}\left\|Z(p)-\ket{\psi_1}\!\bra{\psi_1}\right\|_1=\frac{1}{2}+\frac{\sqrt{2p(1-p)}}{4}.
\end{align}
This lower bound is shown in Fig.~\ref{fig:guess}. Note that it exceeds the tight bounds for $p_{\textrm{NFP}}$ for $M=2$, meaning that, unsurprisingly, Alice can better detect cheating if she can tolerate occasionally falsely accusing Bob. We leave for future investigation whether it is possible to improve upon this lower bound for $p_{\textrm{dist}}$, the lower bounds on $p_{\textrm{dist}}$ for longer messages, and the related question of whether the upper bound of \cref{sec:noGo} is achievable or not.

\section{Acknowledgments}
We're grateful to Guang Ping He for making us aware of past work in this subject and for valuable insights. 
We would particularly like to thank Maxwell Schleck for many helpful suggestions.
We'd also like acknowledge many people who discussed this work or provided feedback:
Matthias Christandl, Paul Hess, Yi-Kai Liu, Mikhail Lukin, Carl Miller, Roee
Ozeri, Anders Sorenson, Jeff Thompson, Stephanie Wehner, Jun Ye, and Terry Rudolph. Kimmel
completed some of this work while at the Joint Center for Quantum Information
and Computer Science at the University of Maryland. Kolkowitz gratefully
acknowledges support from NIST, JILA, CU Boulder, and the NRC postdoctoral
program.

\bibliography{WillBibv6}

\begin{thebibliography}{10}

\bibitem{HistoryofSeals}
Malati~J Shendge.
\newblock The use of seals and the invention of writing.
\newblock {\em Journal of the Economic and Social History of the Orient/Journal
  de l'histoire economique et sociale de l'Orient}, pages 113--136, 1983.

\bibitem{QSeal1}
H~Bechmann-Pasquinucci.
\newblock Quantum seals.
\newblock {\em International Journal of Quantum Information}, 1(02):217--224,
  2003.

\bibitem{QSeal2}
Sudhir~Kumar Singh and R~Srikanth.
\newblock Quantum seals.
\newblock {\em Physica Scripta}, 71(5):433, 2005.

\bibitem{QSeal3}
Guang-Ping He.
\newblock Upper bounds of a class of imperfect quantum sealing protocols.
\newblock {\em Physical Review A}, 71(5):054304, 2005.

\bibitem{QSeal3a}
Helle Bechmann-Pasquinucci, Giacomo~Mauro D'Ariano, and Chiara Macchiavello.
\newblock Impossibility of perfect quantum sealing of classical information.
\newblock {\em International Journal of Quantum Information}, 3(02):435--440,
  2005.

\bibitem{QSeal4}
Guang~Ping He.
\newblock Quantum bit string sealing.
\newblock {\em International Journal of Quantum Information}, 4(04):677--687,
  2006.

\bibitem{QSeal4b}
HF~Chau.
\newblock Insecurity of imperfect quantum bit seal.
\newblock {\em Physics Letters A}, 354(1-2):31--34, 2006.

\bibitem{QSeal5}
HF~Chau.
\newblock Quantum string seal is insecure.
\newblock {\em Physical Review A}, 75(1):012327, 2007.

\bibitem{QSeal6}
Guang~Ping He.
\newblock Comment on ``{Q}uantum string seal is insecure''.
\newblock {\em Physical Review A}, 76(5):056301, 2007.

\bibitem{QSeal7}
HF~Chau.
\newblock Reply to ``{C}omment on `{Q}uantum string seal is insecure'''.
\newblock {\em Physical Review A}, 76(5):056302, 2007.

\bibitem{QSeal8}
Masaki Nakanishi, Seiichiro Tani, and Shigeru Yamashita.
\newblock An information-theoretic security analysis of quantum string sealing.
\newblock In {\em Proceedings of the 6th WSEAS international conference on
  Information security and privacy}, pages 30--35. World Scientific and
  Engineering Academy and Society (WSEAS), 2007.

\bibitem{Qbitcommit}
Dominic Mayers.
\newblock Unconditionally secure quantum bit commitment is impossible.
\newblock {\em Physical Review Letters}, 78(17):3414, 1997.

\bibitem{Qbitcommit2}
Hoi-Kwong Lo and Hoi~Fung Chau.
\newblock Is quantum bit commitment really possible?
\newblock {\em Physical Review Letters}, 78(17):3410, 1997.

\bibitem{BitCommitRudolph}
Robert~W Spekkens and Terry Rudolph.
\newblock Degrees of concealment and bindingness in quantum bit commitment
  protocols.
\newblock {\em Physical Review A}, 65(1):012310, 2001.

\bibitem{goldwasser2008one}
Shafi Goldwasser, Yael~Tauman Kalai, and Guy~N Rothblum.
\newblock One-time programs.
\newblock In {\em Annual International Cryptology Conference}, pages 39--56.
  Springer, 2008.

\bibitem{sealVcommit}
Guang~Ping He.
\newblock Quantum bit commitment is weaker than quantum bit seals.
\newblock {\em The European Physical Journal D}, 68(5):132, 2014.

\bibitem{liu2014building}
Yi-Kai Liu.
\newblock Building one-time memories from isolated qubits.
\newblock In {\em Proceedings of the 5th conference on Innovations in
  theoretical computer science}, pages 269--286. ACM, 2014.

\bibitem{wilde2013quantum}
Mark~M Wilde.
\newblock {\em Quantum information theory}.
\newblock Cambridge University Press, 2013.

\bibitem{winter1999coding}
Andreas Winter.
\newblock Coding theorem and strong converse for quantum channels.
\newblock {\em IEEE Transactions on Information Theory}, 45(7):2481--2485,
  1999.

\bibitem{Watrous2018}
John Watrous.
\newblock {\em The theory of quantum information}.
\newblock Cambridge University Press, 2018.

\bibitem{Kraus83}
K.~Kraus, A.~B{\"o}hm, J.~D. Dollard, and W.~H. Wootters, editors.
\newblock {\em States, Effects, and Operations Fundamental Notions of Quantum
  Theory}, volume 190 of {\em Lecture Notes in Physics, Berlin Springer
  Verlag}, 1983.

\bibitem{braunstein1988quantum}
Samuel~L Braunstein and Carlton~M Caves.
\newblock Quantum rules: An effect can have more than one operation.
\newblock {\em Foundations of Physics Letters}, 1(1):3--12, 1988.

\bibitem{peres2006quantum}
Asher Peres.
\newblock {\em Quantum theory: concepts and methods}, volume~57.
\newblock Springer Science \& Business Media, 2006.

\bibitem{nielsen2002quantum}
Michael~A Nielsen and Isaac Chuang.
\newblock Quantum computation and quantum information, 2002.

\bibitem{Hoeff63}
Wassily Hoeffding.
\newblock Probability inequalities for sums of bounded random variables.
\newblock {\em Journal of the American Statistical Association},
  58(301):13--30, 1963.

\bibitem{blaquiere1987information}
A~Blaquiere, S~Diner, and G~Lochak.
\newblock Information complexity and control in quantum physics.
\newblock In {\em Proceedings of the 4th International Seminar on Mathematical
  Theory of Dynamical Systems and Microphysics Udine. Vienna: Springer}.
  Springer, 1987.

\end{thebibliography}
\bibliographystyle{unsrt}
\end{document}